\documentclass{article}
\usepackage{graphicx}
\usepackage{amsthm}
\usepackage{amsfonts}
\usepackage{amsmath}
\usepackage[pdfstartview=FitH,%
bookmarksnumbered=true,bookmarksopen=true,%
colorlinks=true,pdfborder=001,citecolor=blue,%
linkcolor=blue,urlcolor=blue]{hyperref}

\newtheorem{thm}{Theorem}
\newtheorem{example}{Example}
\newtheorem{lemma}{Lemma}
\newtheorem{prop}{Proposition}
\newtheorem{cor}{Corollary}

\newtheorem{remark}{Remark}

\begin{document}

\title{$K$-anonymous Signaling Scheme}
\author{Binyi Chen \\
Shanghai Jiaotong University
\and 
Tao Qin \\
Microsoft Research Asia
\and 
Tie-Yan Liu \\
Microsoft Research Asia}
\date{September, 2013}
\maketitle
\begin{abstract}
  We incorporate signaling scheme into Ad Auction setting, to achieve better welfare and revenue while protect users' privacy. We propose a new \emph{$K$-anonymous signaling scheme setting}, prove the hardness of the corresponding welfare/revenue maximization problem, and finally propose the algorithms to approximate the optimal revenue or welfare.

\end{abstract}

\section{Introduction}
In real ad Exchange market, auctioneer and bidders will have long term business. On the one hand, leaking impression's information will reduce bidders' competition and thus lose revenue. For example, advertisers usually favor a certain category of users, companies selling football kits only wish to advertise to soccer fans. When the search engine informs the advertiser that the user is a soccer fan, only football companies are willing to bid, which leads to a thinner market and less revenue for the search engine. On the other hand, this will also do harm to users' privacy protection, which is a highly important issue in Internet business. Therefore, it is important to protect users' privacy and hide information of coming impressions.

In order to raise the revenue while protecting privacy, we can inform that the user belongs to a ¡®ball game¡¯ category irrespective of his interest in soccer, basketball, or tennis. Most of the sports companies will bid for this user, leading to higher revenue for the search engine. In the mathematical model, we constrained the support (soccer, basketball, tennis, etc) of the categories (ball game) to be larger than K and our objective is to build up the type (soccer) to category (ball game) map that enables the search engine to achieve the best welfare/revenue.

In section 2, we propose the specific setting of the problem. In section 3, we firstly prove that it is NP-hard to solve K-anonymous signaling welfare maximization problem and it cannot be approximated in factor less than $e/(e-1)$. Then we give a $2e/(e-1)$ approximation algorithms. Finally, we propose several special settings that can be solved optimally in polynomial time. In section 4, we prove the NP-Hardness of revenue maximization, and give a method to transfer welfare maximization approximation result into revenue approximation. In section 5, we note the possible future work.

\section{Problem Formulation}

There is one impression coming for auction, the impression belongs to one of $m$ category, the prior probability that the impression belongs to the $i$th category is  $p_i$. The $i$th bidder of $n$ bidders values the $j$th category as $v_{ij}$. Auctioneer will broadcast a signal $S$ after knowing the impression's category and raise a second price auction. The map between categories and signals is public to bidders. Denote $\phi(j, S)$ as the probability that auctioneer will broadcast signal $S$ when category is $j$. In this setting, $\phi(j,S) \in \{0,1\}$. Finally, auctioneer wants to construct the best map that can extract the most revenue or social welfare, while satisfies that each signal's support(categories that will broadcast this signal) $sup_{\phi,S} = \{j:\phi(j,S)=1\}$ is equal or larger than $K$.

Without \emph{K-constraint}, the model is a typical \emph{pure signaling scheme} mentioned in \cite{ec12a}.  For a signal $S$, the probability of broadcasting $S$ is $\sum_{j}p_{j}\phi(j,S)$. Given the auctioneer broadcasting $S$, the probability that the category is $j$ is $Pr[j|S] = p_{j}\phi(j,S)/(\sum_{j'}p_{j'}\phi(j',S)$. The expected welfare achieved by broadcasting $S$ is
\[
\max_{i} \{\mathbf{E}[v_i | S] \} = max_{i} \{\frac{\sum_{j}v_{i,j}p_{j}\phi(j,S)}{\sum_{j}p_{j}\phi(j,S)} \}
\]
Therefore, the objective is to find a map $\phi$ that satisfies \emph{K-constraint} to maximize
\[
\sum_{S}Pr[S]\max_{i} \{\mathbf{E}[v_i | S] \} = \sum_{S}\max_{i} \{\sum_{j}(v_{i,j}p_{j})\phi(j,S) \}
\]
We use $V_{i,j}$ to describe $v_{i,j}p_{j}$ for simplicity later. The revenue maximization problem is similar by replacing maximum with the second maximum since second price auction is truthful in this setting.

We observe that this problem is equivalent to the bundling scheme problem: there are $m$ items for sale, bidders have different value to different item. The task is to partition items into several bundles, each bundle has size at least $K$. Then auctioneer sells bundles separately by second price auction. The goal of auctioneer is to maximize social welfare/revenue.

\section{Welfare Maximization}
In this section we firstly prove the NP-Hardness of welfare maximization problem and prove it cannot be approximated in factor less than $e/(e-1)$. Then we give a $2e/(e-1)$ approximation algorithms. Finally, we propose several special settings that can be solved optimally in polynomial time.

\subsection{Hardness Results}

When both number of signals used and bundle size have no constraint, welfare maximization problem is trivial by giving category $j$ impression to bidder who values it most. However, when signal number is constrained \emph{(Cardinality Constrained Signaling Problem)}, or each signal's support size has to be no less than $K$ \emph{(K-anonymous Signaling Problem)}, the problem becomes NP-Hard.

We prove the NP-Hardness of K-anonymous welfare-maximization problem by a reduction from \cite{dughmi}.

\begin{prop} \cite{dughmi}
There is no polynomial-time $c$-approximation algorithm for welfare-maximization with known valuation, when signal number is constrained, for any constant $c<\frac{e}{e-1}$, unless $P=NP$.
\end{prop}
\begin{proof}
The proof is directly from \cite{dughmi}.
\end{proof}

\begin{thm}
There is no polynomial-time $c$-approximation algorithm for welfare-maximization with known valuation and K-anonymous constraint, for any constant $c<\frac{e}{e-1}$, unless $P=NP$.
\end{thm}
\begin{proof}
Prove by contradiction. Assume there is a $c$-approximation algorithm for $c<\frac{e}{e-1}$.

Given any signal number-constraint problem instance $I$ with input $m, S, V_{ij}$, $m$ is the number of items, $S$ is signal number constraint, $V_{ij} =v_{i,j}p_{j}$. W.l.o.g, we assume $V_{ij}$ are all integer values by scaling (for the proof of \cite{dughmi} still applies when $v_{i,j}p_{j}$ are rational number). We denote $OPT$ as the optimal solution. We can assume there are exactly $S$ bundles(signals) in $OPT$ as $b_1,b_2,...,b_S$ since we can split bundles without decreasing welfare if bundle (signal) number is less than $S$.

We create a new instance for K-anonymous signaling $I_K$. Set $K=\lceil \frac{m-S}{2}+1 \rceil$, items are $m$ old items plus $K \cdot S-m+K-1$ new items, bidders are $n$ old bidders plus $K \cdot S-m+K-1$ new bidders. The $i$th new bidder only positively values the $i$th new item with $\frac{1}{S}$, and values the other items zero.

Notice $K \cdot S-m+K-1$=$K \cdot (S+1)-m-1$ is non-negative for any $S>0$. Since:
\[
K \cdot (S+1)-m-1 \geq (\frac{m-S}{2}+1) (S+1) - m - 1 \geq m-S+(S+1)-m-1 = 0
\]

We denote $OPT^{'}$ as the optimal solution of instance $I_K$.

\begin{lemma} \label{lem1}
$\sum_{b_i \in OPT} (b_{i}-K)^{+} < K$
\end{lemma}
\begin{proof}
Assume left-hand side $\geq K$, let $c$ the number of bundles that exceeds K items. Then there are $S-c$ bundles with at least one item but less than $K+1$ items.

When $c=0$ the lemma holds obviously, when $c>0$ we have
\begin{align}
      m &\geq c \cdot K+K+S-c \\
        &\geq c \cdot \frac{m-S}{2}+\frac{m-S}{2}+c+1+s-c \\
        &\geq m-S+c+1+s-c = m+1
\end{align}
Contradiction. The second inequality holds by substituting $K=\frac{m-S}{2}+1$, the third inequality holds for $c>0$.
\end{proof}

\begin{lemma} \label{lem2}
$OPT^{'} \geq OPT$
\end{lemma}
\begin{proof}
Since the total number of items are $K \cdot S+K-1$ and $\sum_{b_i \in OPT} (b_{i}-K)^{+} < K$, we can fill bundles $b_i < K$ by new items and fix bundles $b_j \geq K$, making each bundle size not smaller than $K$. Thus we can achieve at least welfare $OPT$.
\end{proof}

Partition $OPT^{'}$ into 2 kinds of bundles $A=\left\{a_1,...,a_{t_1}\right\}$, $C=\left\{c_1,...,c_{t_2}\right\}$. A bundle is in $C$ if and only if all of its items are new items. $t_1+t_2 \leq S$ since the total number of items is $K \cdot S+K-1$. Winners in $A$ are all old bidders since new bidder can value a bundle at most $1/S$ but at least one old bidder value the bundle $\geq 1$. On the other hand, each bundle in $C$ is valued at most $1/S$, and $t_2 < S$, thus sum of welfare in $C$ is less than $(S-1)/S$.

Remove all new items in $OPT^{'}$, we will obtain a feasible solution $OPT^{''}$ for $I$. We only lose welfare in $C$ which is at most $(S-1)/S$.Thus $OPT^{''}>OPT-1$. Since $V_{ij}$ are all integers, we have $OPT^{''} = OPT$. Therefore, we obtain a $c$-approximation algorithm for welfare-maximization with signal number constraint, for $c<e/(e-1)$. Contradiction.

\end{proof}

\subsection{$2e/(e-1)$ Approximation Algorithm}

We provide a $2e/(e-1)$ Approximation Algorithm by incorporate result from \cite{dughmi}

\begin{prop} \cite{dughmi}
For cardinality constrained signaling with known valuations, there is a randomized, polynomial-time, $e/(e-1)$-approximation algorithm for computing the welfare-maximizing signaling scheme.
\end{prop}
\begin{proof}
The proof is directly from \cite{dughmi}.
\end{proof}

\begin{thm}
For K-anonymous signaling with known valuations, there is a randomized, polynomial-time, $2e/(e-1)$-approximation algorithm for computing the welfare-maximizing signaling scheme.
\end{thm}
\begin{proof}
Given a K-anonymous signaling instance $I_K$, we set $S$ as $\lfloor m/K \rfloor$ and solve cardinality constrained signaling problem, without K-anonymous constraint.

After obtain the solution $ALG$, assume set $A$ has $t$ signals, each of which has less than $K$ items: $A=\left\{a_1,a_2,...,a_t\right\}$, set $B$ has $\lfloor m/K \rfloor-t$ signals, each of which has no less than $K$ items: $B=\left\{b_1,b_2,...,b_{\lfloor m/K \rfloor -t}\right\}$.

If welfare of $B$ is no less than half of the total welfare, then we are done. Otherwise, $A$'s welfare will be larger than half of the total value. We can fill each $a$ by $B$'s items, since items left is no less than $m-(a_1+...+a_t) \geq k \cdot t-(a_1+...+a_t)$. Then we can obtain $A$'s welfare in a feasible solution, it is also a $2e/(e-1)$ approximation.
\end{proof}

We show the solution of algorithm above can be only half of the optimum.
\begin{example}
There are $K^2+K$ items and $K+2$ bidders. Bidder $A$ values each of the first $K^2$ items for $1/K$. And bidder $B_i$ values the $K^2+i$th item for 1 and others zero. Bidder $C$ value each of the last $K$ items for $1-\frac{\epsilon}{K}$.

The optimal solution's welfare is $2K-\epsilon$, by giving first $K^2$ items to $A$ and last $K$ items to $C$. But the algorithm's best welfare can be at most $K+1$, since $B_i$ will win in $ALG$. The gap approaches to 2 as K goes to infinity.
\end{example}

\subsection{Special Cases}

Given some extra assumption, we can calculate the optimal solution for welfare maximization in polynomial time.

\subsubsection{Signals Number Constraint}
We constrain the number of signals being used as a constant, and propose an algorithm that optimally solve the problem in polynomial time. It consists of 3 steps as below:
\begin{enumerate}
  \item Enumerate signals number and signals' winners. This only cost $O(n^{c})$ time.
  \item Remove bidders who won no signal, making each item chosen by bidder who favors it most. After this step, the $i$th bidder get $k_i$ items. This scheme achieve the best social welfare, yet some winner may win less than k items.
  \item Use cost flow technique to find the best allocation which satisfy the K-anonymous constraint. The graph is shown in the figure 1, arcs are described as $(l,u,c)$, $l$ is the lower bound of flow, $u$ is the upper bound, $c$ is the cost per flow. $w_i$ and $w_{i^{'}}$ correspond to the $i$th bidder, $item_j$ corresponds to the $j$th item. Arc from $st$ to $w_i$ represents the initial number of items of bidder $i$ is $k_i$. Arc from $w_{i}^{'}$ to $end$ represents the final items of bidder $i$ should be no less than $k$. Arc from $w_{i}^{'}$ to $item_j$ represent $i$ can give its $item_j$ to others, leading a welfare decrease of $v_{ij}$. Arc from $item_k$ to $w_i$ represents that bidder $i$ can gain $item_k$ from others, leading a welfare increase of $v_{ij}$. By solving the minimum cost feasible flow problem \cite{flow}, we can obtain final optimal allocation.
\end{enumerate}

\begin{figure}
  \centering
  \includegraphics[width=0.8\textwidth]{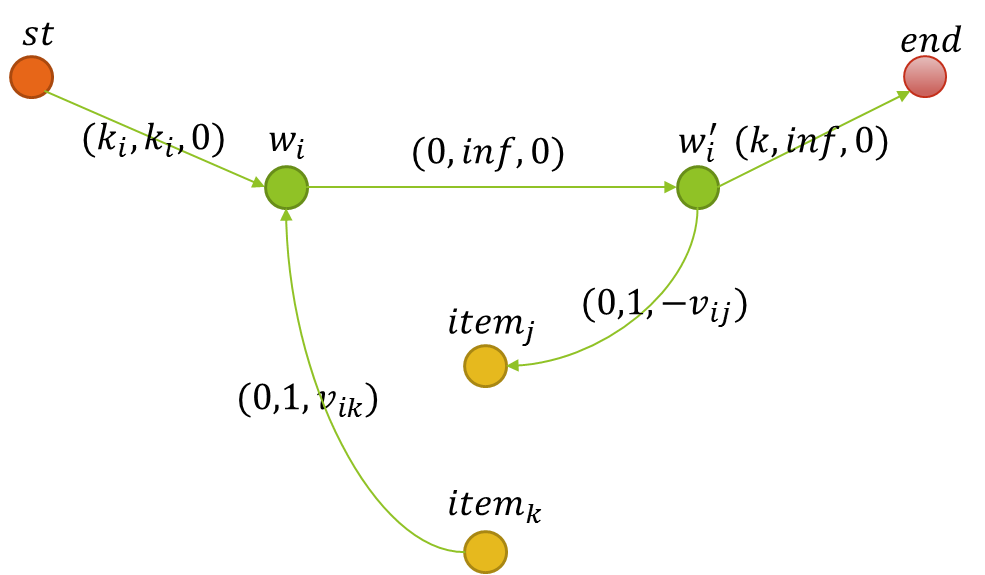}\\
  \caption{Cost Flow}\label{fig:1}
\end{figure}
\subsubsection{Value Constraint}
In this subsection, we propose a setting that the value of bidder $i$ to item $j$ $V_{ij}$ can be expressed as $p_{i} \cdot q_{j}+b_i$. A dynamic programming approach is given that solve it in polynomial-time.

\begin{lemma}
Given any instance of K-anonymous signaling problem, there always exists an optimal signaling scheme that, sort signal winners $\left\{w_1,w_2,...,w_t\right\}$ by $p_i$ in non-decreasing order. $\forall i<j, p_i \neq p_j$, for any category $t_i$ points to $w_i$'s signal, and any category $t_j$ points to $w_j$'s signal, $q_{t_i} \geq q_{t_j}$.
\end{lemma}
\begin{proof}
Given any optimal signaling scheme, assume there exists $i<j,p_i \neq p_j, q_k>q_l$ such that $w_i$ possesses $q_l$ and $w_j$ possesses $q_k$. After swap $q_k$ and $q_l$. The welfare change is
\begin{align}\label{}
  p_{i}q_{k}+p_{j}q_{l} - (p_{i}q_{l}+p_{j}q_{k}) &= p_{i}(q_{k}-q_{l}) - p_{j}(q_{k}-q_{l}) \\
    &= (p_{i}-p_{j})(q_{k}-q_{l}) \geq 0
\end{align}
Finally we obtain an optimal signaling scheme that satisfies the lemma by repeatedly swapping.
\end{proof}

Thus we can sort bidders(items) with respect to $p_i$($q_j$). Building state $F(i,j)$, which means the best welfare for bidders $\left\{1,2,...,i\right\}$, items $\left\{1,2,...,j\right\}$.

\begin{equation}
F(i,j) = \max \begin{cases}
                F(i-1,j)\\
                F(i,j-1)\\
                F(i-1,j-K_1)+K_1 \cdot b_i + (\sum_{t=j-K_1+1}^{j} q_t)p_i
                \end{cases}
\end{equation}
$K_1$ are the enumerated number of items that bidder $i$ will win, $K_1 \geq K$. The optimal solution is $F(n,m)$. The complexity is $O(nm^2)$.

\section{Revenue Maximization}

\subsection{Hardness Result}
\begin{thm}
Even when the signals number is constant, it is still NP-hard to solve the revenue maximization problem in K-anonymous constraint setting.
\end{thm}

\begin{lemma}
It is NP-hard to partition $2n$ integers $\{x_1,...,x_{2n} \}$ into two sets with size $n$, such that the two subsets' element sum are the same.
\end{lemma}
\begin{proof}
Without size $n$ constraint, it is a classical NP-Complete problem(Subset Sum \cite{subsetsum}) that we want to find a subset whose sum is $\frac{\sum_{i}x_i}{2}$.

Proof by contradiction, assume that the new problem has polynomial time solution. For any Subset Sum problem instance with $n$ integers, we can add $n$ zeros and transfer it into the new problem. By finding two sets with size $n$ and the sum are the same, we can remove zeros and obtain a solution for Subset Sum, contradiction.
\end{proof}

\begin{proof}
Assume the revenue optimization can be polynomially solved. For any instance of Subset Sum Problem with Size $K-1$ constraint(partition $2K-2$ integers into two $K-1$-size sets with the same sum), SSPS. We denote the sum of all integers as $W$ and raise an auction: There are 3 bidders, $2K$ items, signal's support size should be no less than $K$. the value matrix shows in table below:

\begin{tabular}{|c|c|c|c|c|c|c|}
\hline
         &             &             & Item  &       &     &         \\
\hline
         &$\frac{W}{2}$&             &       &       &     &         \\
\hline
  Bidder &             &$\frac{W}{2}$&       &       &     &         \\
\hline
         &             &             & $x_1$ & $x_2$ & ... & $x_{2K-2}$ \\
\hline
\end{tabular}

Signal number can be at most 2. When signal number is one, the revenue can be at most $\frac{W}{2}$.

When signal number is two, if there exists a solution $\{S_1, S_2 \}$ of SSPS, we can bundle item $ \{1 \cup S_1 \} $ and $ \{2 \cup S_2 \} $, achieve revenue $W$. Since any bundle can achieve value at most $W/2$, this is optimal.

When there is no solution for SSPS, given any signaling scheme with 2 signals, if items $1,2$ are not in the same signal, the revenue is strictly less than $W$. When $1,2$ are in the same signal $s_1$, $s_1$'s revenue is $W/2$ and $s_2$'s revenue is strictly less than $W/2$. Therefore, the optimal revenue is strictly less than $W$.

From all above, the optimal revenue of the auction is $W$ if and only if there exists a solution of SSPS, Contradict with the fact that Subset Sum problem is NP-hard.

\end{proof}
\begin{remark}
This proof's rough thought is coincidentally similar to the proof of NP-Hardness of revenue maximization when K-anonymous constraint does not exists. \cite{ghosh}
\end{remark}

\subsection{Approximation Approach}
\cite{dughmi} independently gives a randomized mechanism for approximating revenue in \emph{cardinality constraint} setting, we independently give an deterministic approximation algorithm for \emph{K-anonymous} setting with similar yet different analysis (Since the valuation is known in prior, an algorithm is sufficient to extract revenue).

\begin{thm}
We can achieve 3$\beta$-approximation in revenue maximization problem with K-anonymous constraint, when there is a $\beta$-approximation for welfare maximization problem.
\end{thm}

\begin{proof}
We prove that we can transfer welfare maximization result into revenue maximization.

Fix two parameters $\beta$ and $\alpha$, we firstly find the $\beta$-approximate optimal signaling scheme $S$ for welfare maximization in K-anonymous constraint. Denote $V_{i}$ as the $i$th bidder's welfare contribution in $S$. Let $V^{*}=\max(V_{i})$, $OPTW$ as the social welfare of scheme $S$, and $OPT$ as the optimal revenue we can achieve, obviously $OPTW \geq OPT$. We first sort bidders according to $V_i$.

\begin{itemize}
  \item If $V^{*} \leq \beta OPTW$, merge the $1$st and the $2$nd bidders' signal, $3$rd and $4$th bidders' bundles, and so on. We can at least achieve $\frac{1-\beta}{2}OPT$ revenue without violate K-anonymous constraint.

  \item We denote $OPTW(-i^{*})$ as the best social welfare we can achieve when the first bidder $i^*$ has been removed(signaling scheme is $S^{'}$). Obviously, $OPTW(-i^{*}) \geq OPT$. We further denote $V'^{*}$ as the max welfare an individual bidder can contribute in $OPTW(-i^{*})$.

      If $V^{*} \geq \beta OPTW$ and $V'^{*} \leq \alpha OPTW(-i^{*})$, In signaling scheme $S^{'}$, we merge the $1$st and the $2$nd bidders' signals, 3rd and 4th bidders' signals, and so on. We can at least achieve $\frac{1-\alpha}{2}OPT$ revenue.

  \item Otherwise, sell all items in one signal, we can at least achieve $\min(\alpha, \beta)OPT$ revenue.
\end{itemize}
Setting $\alpha=\beta=\frac{1}{3}$, we achieve a $\frac{1}{3}$ competitive ratio.
\end{proof}

\begin{cor}
For K-anonymous signaling with known valuations, there is a randomized, polynomial-time, $6e/(e-1)$-approximation algorithm for computing the revenue-maximizing signaling scheme.
\end{cor}

\section{Future Work}
An $\alpha$-approximation algorithm for welfare maximization problem that $\frac{e}{e-1} \leq \alpha < \frac{2e}{e-1}$ is the first direction that we may keep working on.

The second future direction is to develop better approximation algorithm for revenue maximization independently, without transferring welfare maximization results.

Finally, our work only focus on signaling schemes that $\phi(j,S) \in \{0,1\}$. In probabilistic setting with mixed signal that mentioned in \cite{ec12a}, new K-anonymous signaling scheme setting may be proposed, and we may get other useful approximation algorithms or even an optimal one.

\end{document}